\documentclass[aps,pra,showpacs,twocolumn,nofootinbib,10pt]{revtex4-1}
\pdfoutput=1

\usepackage{graphicx}
\usepackage{amsmath}
\usepackage{amssymb}
\usepackage{mathrsfs}
\usepackage{amsthm}
\usepackage{bm}
\usepackage{url}
\usepackage[T1]{fontenc}
\usepackage{csquotes}
\MakeOuterQuote{"}

\newtheoremstyle{note}
  {\topsep/2}               
  {\topsep/2}               
  {}                      
  {\parindent}            
  {\itshape}              
  {.}                     
  {5pt plus 1pt minus 1pt}
  {}

\theoremstyle{note}
\newtheorem{theorem}{Theorem}
\newtheorem{lemma}{Lemma}
\newtheorem{conjecture}{Conjecture}

\theoremstyle{definition}

\theoremstyle{remark}
\newtheorem{remark}{Remark}

\newcommand{\vecrm}[1]{\mathbf{#1}}
\newcommand{\mrm}[1]{\mathrm{#1}}
\providecommand{\tr}{\operatorname{tr}}

\newcommand{\rep}{\mathrel{\widehat{=}}}

\providecommand{\rmi}{\mathrm{i}}
\providecommand{\rme}{\mathrm{e}}

\providecommand{\rmd}{\mathrm{d}}
\newcommand{\rmT}{\mathrm{T}}

\newcommand{\bbF}{\mathbb{F}}

\newcommand{\be}{\begin{equation}}
\newcommand{\ee}{\end{equation}}
\newcommand{\ba}{\begin{align}}
\newcommand{\ea}{\end{align}}

\def\<{\langle}  
\def\>{\rangle}  

\newcommand{\Sp}[2]{\mrm{Sp}(#1,#2)}

\newcommand{\ASp}[2]{\mrm{ASp}(#1,#2)}
\newcommand{\SL}[2]{\mrm{SL}(#1,#2)}

\newcommand{\PSL}[2]{\mrm{PSL}(#1,#2)}
\newcommand{\ASL}[2]{\mrm{ASL}(#1,#2)}

\newcommand{\hw}{D}

\newcommand{\phw}{\overline{D}}

\newcommand{\Cli}{\mathrm{C}}

\newcommand{\pc}{\overline{\mathrm{C}}}

\newcommand{\Clir}{\mathrm{C}_{\mrm{r}}}
\newcommand{\pcr}{\overline{\mathrm{C}}_{\mrm{r}}}

\def\eqref#1{\textup{(\ref{#1})}}  
\newcommand{\eref}[1]{Eq.~\textup{(\ref{#1})}}

\newcommand{\thref}[1]{Theorem~\ref{#1}}
\newcommand{\Thref}[1]{Theorem~\ref{#1}}

\newcommand{\lref}[1]{Lemma~\ref{#1}}

\newcommand{\cref}[1]{Conjecture~\ref{#1}}
\newcommand{\Cref}[1]{Conjecture~\ref{#1}}

\newcommand{\rcite}[1]{Ref.~\cite{#1}}
\newcommand{\rscite}[1]{Refs.~\cite{#1}}

\begin{document}
	\title{Mutually unbiased bases as minimal Clifford covariant 2-designs}
	\author{Huangjun Zhu}
	\email{hzhu@pitp.ca}
	\affiliation{Perimeter Institute for Theoretical Physics, Waterloo, Ontario N2L 2Y5, Canada}
	
	\pacs{03.67.-a, 02.10.De, 03.65.-w}



\begin{abstract}
Mutually unbiased bases (MUB) are interesting for various reasons. The most attractive example of (a complete set of) MUB is
the one constructed by Ivanovi\'c as well as Wootters and Fields, which is referred to as the canonical MUB. Nevertheless, little is known about anything that is unique to this MUB. We show that the canonical MUB in any prime power dimension is uniquely determined by  an extremal orbit of the (restricted) Clifford group except in dimension 3, in which case the orbit defines a special symmetric informationally complete measurement (SIC), known as the Hesse SIC. Here the extremal orbit is the one with the smallest number of pure states. Quite surprisingly, this characterization does not rely on any concept that is related to bases or unbiasedness. As a corollary, the canonical MUB is the unique minimal 2-design covariant with respect to the Clifford group except in dimension 3. In addition, these MUB provide an infinite family of highly symmetric frames and positive-operator-valued measures (POVMs), which are of independent interest.

\end{abstract}
\date{\today}
\maketitle

\emph{Mutually unbiased bases} (MUB) are useful to a number of research areas, such as quantum kinematics, quantum state tomography, and quantum cryptography. They are also interesting because of their connections with discrete Wigner functions, symmetric informationally complete measurements (SICs for short), and generalized Bell states; see \rcite{DurtEBZ10} for a review.
In a $d$-dimensional Hilbert space,
two bases  are mutually unbiased if the transition probability across their basis states are all equal to $1/d$.
Each MUB contains at most
$d+1$ bases~\cite{WootF89}; the MUB is complete if the upper bound is attained. In the rest of the paper by a MUB we shall mean such a complete set.

In each prime dimension, MUB was constructed by Ivanovi\'c \cite{Ivan81} motivated by a state estimation problem, which was  generalized to any prime power dimension by Wootters and Fields \cite{WootF89}. This MUB, referred to as the canonical MUB henceforth,  is interesting not only to the physics community but also to a number of other research areas, such as signal processing \cite{Allt80} and operator algebra \cite{Popa83}. Although many different MUB were constructed thereafter, the canonical MUB has been the focus of most literature. Almost all MUB known so far, including the canonical MUB, can be equivalently constructed from \emph{stabilizer states} \cite{Gott97the,LawrBZ02, BandBRV02,GodsR09, Kant12,Zhu15Sh}, which are simultaneous eigenstates of Heisenberg-Weyl (HW)  displacement operators (also known as generalized Pauli operators). The symmetry group of any such MUB is a subgroup of the full Clifford group, the normalizer of the HW group \cite{CormGGP06, Kant12, Zhu15Sh}.
The canonical MUB is known to be covariant with respect to the restricted Clifford group, an important subgroup of the full Clifford group \cite{ApplDF14,Appl09P, Kant12}. On the other hand, little is known about anything that is unique to the canonical MUB. One may wonder whether its prominent status is due to historical reasons or simply because  other MUB are not familiar to most researchers.

In this paper,
we show that the canonical MUB is uniquely determined by a special orbit  of the restricted Clifford group, the orbit with the smallest number  of pure states. The only exception occurs in dimension 3 in which case the  orbit defines a peculiar SIC, known as the Hesse SIC \cite{Zaun11, Appl05, Zhu10, Zhu14S,Zhu15P}. Quite surprisingly, this characterization of the canonical MUB does not refer to any concept related to bases or unbiasedness,  which are the starting points of all other definitions we are aware of.
As a corollary, the canonical MUB is the unique minimal 2-design covariant with respect to the restricted Clifford group except in dimension 3. In addition, our study shows that the canonical MUB defines a highly symmetric frame and positive-operator-valued measure (POVM) \cite{BrooW13,CasaK13,SlomS14}, which are of  interest in signal processing and quantum information science.

In prime dimension $p$, the HW group $\hw$ is generated by the phase operator $Z$ and cyclic shift operator $X$ (together with scalar $\rmi$ when $p=2$),
\begin{equation} \label{eq:HW}
Z|r\rangle=\omega^r|r\rangle, \qquad X|r\rangle=
|r+1\rangle,
\end{equation}
where $\omega=\mathrm{e}^{2\pi \mathrm{i}/p}$,
$r\in \bbF_p$, and $\bbF_p$  is the field
of integers modulo $p$. The (multipartite) HW group in prime power dimension $q=p^n$ is usually defined as the tensor power of $n$ copies of the HW group in dimension $p$.
A stabilizer basis is the common eigenbasis of a maximal abelian subgroup of the HW group, where an abelian subgroup is  maximal if it has order $q$ modular phase factors. Two stabilizer bases constructed in this way are mutually unbiased if and only if the corresponding maximal abelian subgroups have trivial intersection \cite{Gott97the,LawrBZ02, BandBRV02, AschCW07, GodsR09, Kant12,Zhu15Sh}.
When $q$ is a prime, there are only $q+1$ stabilizer bases, so the  stabilizer MUB is unique.
In general,  many different  MUB can be constructed from stabilizer bases \cite{Kant12}, but
most existing literature has focused on a special example, namely, the canonical MUB.

 Before introducing the  canonical MUB, it is convenient to adopt another equivalent definition of the HW group~$\hw$ that  enjoys  nice properties of finite fields. Let $\bbF_q$ be the finite field with $q$ elements, whose elements label the computational basis. Define
\begin{equation} \label{eq:HWmul}
X_u |x\rangle=|x+u\rangle, \quad Z_u |x\rangle=\omega^{\tr (ux)} |x\rangle,
\end{equation}
where $u, x\in \bbF_q$ and "$\tr$" denotes the field theoretic trace \cite{Gros06, Appl09P}. Up to phase factors, the elements of the HW group can be labeled by vectors in $\bbF_q^2$,
\begin{equation}
D_{\vecrm{u} }:=\tau^{\tr(u_1 u_2)} X_{u_1} Z_{u_2},\quad \tau=-\rme^{\rmi \pi/p},
\end{equation}
which satisfy
\begin{equation}
D_{\vecrm{u}}D_{\vecrm{v}} D_{\vecrm{u}}^\dag D_{\vecrm{v}}^\dag  =\omega^{\langle \vecrm{u},\vecrm{v}\rangle},
\quad \langle \vecrm{u},\vecrm{v}\rangle:= \tr(u_2v_1-u_1v_2).
\end{equation}
It is straightforward to verify that the $q$ displacement operators labeled by vectors   on each ray in $\bbF_q^2$ commute with each other and thus generate a maximal abelian subgroup of the HW group. In  this way, each ray in $\bbF_q^2$ determines a maximal abelian subgroup and also a stabilizer basis. Since any two such groups have trivial intersection,  the $q+1$ stabilizer bases thus constructed are mutually unbiased; actually, they form the canonical MUB \cite{LawrBZ02, BandBRV02,GodsR09, Kant12}.

The full Clifford group $\Cli$   is composed of all unitary operators that map displacement operators to displacement operators up to phase factors \cite{BoltRW61I,BoltRW61II,Gott97the}. Any Clifford unitary $U$ induces an $\bbF_p$-linear transformation on $\bbF_q^2$ (identified with $\bbF_p^{2n}$) that labels the displacement operators, that is,
\begin{equation}
UD_{\vecrm{u} } U^\dag=\rme^{\rmi \varphi(\vecrm{u})} D_{f(\vecrm{u}) },
\end{equation}
where $\varphi(\vecrm{u})$ is a phase factor of no concern here, and $f$
satisfies $f(\vecrm{u}+\vecrm{v})= f(\vecrm{u})+f(\vecrm{v})$ as well as  $f(a\vecrm{u})=a f(\vecrm{u})$ for all $a\in \bbF_p$. Note that $D_{\vecrm{u}}U$ for all $\vecrm{u}\in \bbF_q^2$ induce the same transformation $f$ as $U$.
Since conjugation preserves the commutation relation, the linear transformation $f$ also belongs to the
symplectic group $\Sp{2n}{p}$ with respect to a suitable symplectic form. Conversely, for any linear transformation  in  $\Sp{2n}{p}$, there exists $q^2$ Clifford unitaries (up to phase factors) that induce the given linear transformation.
The quotient group $\pc/\phw$ ($\overline{G}$ denotes the collineation group of $G$, that is, $G$ modulo phase factors) can be identified with the symplectic group $\Sp{2n}{p}$.
When $p$ is odd, $\pc$ is also isomorphic to
the affine symplectic group $\ASp{2n}{p}=\Sp{2n}{p}\ltimes \bbF_p^{2n}$ \cite{BoltRW61I,BoltRW61II}.

The restricted Clifford group $\Clir$ (coinciding with the full Clifford group when $q$ is prime) is the subgroup of $\Cli$ obtained if we require that $f$ is in addition $\bbF_q$-linear, that is, $f(a\vecrm{u})=a f(\vecrm{u})$ for all $a\in \bbF_q$ \cite{Gros06,Appl09P}. This  requirement implies that $f(\vecrm{u})=F\vecrm{u}$ with $F$ belonging  to $\Sp{2}{q}\simeq\SL{2}{q}$; that is,
\begin{equation}
F=\begin{pmatrix}
\alpha &\beta \\ \gamma &\delta
\end{pmatrix},\quad \alpha\delta-\beta\gamma=1.
\end{equation}
Accordingly, the quotient group $\pcr/\phw$ can be identified with $\SL{2}{q}$, which has order $q(q^2-1)$, and  the restricted Clifford group $\pcr$ has order $q^3(q^2-1)$. When $q>3$, the group $\SL{2}{q}$ can be generated by the following two elements \cite{Tayl06}:
\begin{equation}
\begin{pmatrix}\nu & 0\\ 0 & \nu^{-1}  \end{pmatrix}, \quad \begin{pmatrix}-1& 1\\ -1 & 0  \end{pmatrix},
\end{equation}
where $\nu$ is a primitive element in $\bbF_q$, that is, a generator of the multiplicative group $\bbF_{q}^*$ composed of nonzero elements in $\bbF_q$.
When $q=2, 3$, the generators can be chosen to be
\begin{equation}
\begin{pmatrix}1 & 1 \\  0& 1  \end{pmatrix}, \quad \begin{pmatrix}0& 1\\ -1 & 0  \end{pmatrix}.
\end{equation}
In any case, the restricted Clifford group can be generated by two Clifford unitaries and a nontrivial displacement operator.

When $q$ is odd, the restricted Clifford group $\pcr$ is isomorphic to the affine special linear group $\ASp{2}{q}\simeq\ASL{2}{q}$ \cite{BoltRW61I,BoltRW61II,Appl09P}. In addition, for each $F\in \SL{2}{q}$,
there exists a Clifford unitary $U_F$ satisfying
\begin{equation}
U_FD_{\vecrm{u} } U_F^\dag =D_{F\vecrm{u} };
\end{equation}
an explicit formula for $U_F$  was derived by Appleby \cite{Appl09P},
\begin{equation}\label{eq:Cliffordunitary}
U_F=\begin{cases}\sum\limits
_{x\in \bbF_q}|\alpha x\rangle\tau^{\tr(\alpha\gamma x^2)}\langle x|, &\beta=0;\\
\frac{1}{\sqrt{q}}\sum\limits
_{x,y\in \bbF_q}|x\rangle\tau^{\tr[\beta^{-1}(\alpha y^2-2xy+\delta
x^2)]}\langle y|, &\beta\neq 0.
\end{cases}
\end{equation}
When $q$ is even, we are not aware of a simple formula for Clifford unitaries except when $q=2$ \cite{Appl05}, but there is an algorithm described in \rcite{DehaM03}.

Let $\{\Pi_j\}$ be the set of projectors corresponding
to basis states in a MUB; the symmetry group of the MUB is the group of all unitary transformations $U$ modulo phase factors that leave the set $\{\Pi_j\}$ invariant; that is, $U\Pi_jU^\dag=\Pi_{\sigma(j)}$ for some permutation $\sigma$ of the indices.
The restricted Clifford group leaves invariant the set of $q+1$ maximal abelian subgroups of the HW group that correspond to the $q+1$ rays in $\bbF_q^2$. Therefore, the symmetry group of the canonical MUB contains the restricted Clifford group. In addition, the canonical MUB is covariant with respect to the restricted Clifford group; by covariance we mean that the group   acts transitively on the states in the MUB (similar definitions apply to SICs, 2-designs, and frames etc.).
What is not so obvious is that  it is also the only MUB with this property. Actually, we can prove an even stronger result, as spelled out in \thref{thm:MinOrbit} below.  It should be noted that the symmetry group of the canonical MUB is the normalizer of
the restricted Clifford group, which is strictly larger than the latter group except  in a prime dimension  \cite{Kant12}.
To be specific, the normalizer has  order $nq^3(q^2-1)$, while the restricted Clifford group has order $q^3(q^2-1)$. This result is related to the fact that $\SL{2}{q}$ has index $n$ in its normalizer within $\Sp{2n}{p}$ \cite{KleiL90}.

To present our main result, we need to introduce a new concept. Given a group $\overline{G}$ of unitary transformations in dimension $d$, the set of all pure states in dimension $d$ form disjoint orbits under the natural action of $\overline{G}$. An orbit of $\overline{G}$ is \emph{characteristic} if its cardinality is the smallest among all orbits. In general, the characteristic orbit may not be unique, but it is unique for many groups of special interest, such as the restricted Clifford group.
\begin{theorem}\label{thm:MinOrbit}
The restricted Clifford group in dimension $q=p^n$ has a unique characteristic orbit, which corresponds to the Hesse SIC when $q=3$ and the canonical MUB otherwise.
\end{theorem}
\begin{remark}
In the case $q=3$, the canonical MUB is the unique orbit of the Clifford group with the second smallest cardinality.
\end{remark}

In addition to providing a simple characterization of the canonical MUB, \thref{thm:MinOrbit} also has other interesting implications.
Recall that a
set of pure states $\{|\psi_j\rangle\}$ is a (projective) \emph{2-design} \cite{Zaun11, ReneBSC04, Scot06, ApplFZ15G} if $\sum_j (|\psi_j\rangle\langle \psi_j|)^{\otimes 2}$ is proportional to the projector onto the symmetric subspace of the bipartite Hilbert space. 2-designs are useful to a number of quantum information processing tasks, such as quantum state estimation \cite{Scot06, ZhuE11, Zhu12the, Zhu14T}.
Any 2-design  $\{|\psi_j\rangle\}$ in dimension $d$ has at least $d^2$ elements, and the lower bound is saturated if and only if the 2-design corresponds to a SIC  \cite{Zaun11, ReneBSC04, Scot06, ApplFZ15G}, that is,
\begin{equation}
|\langle \psi_j|\psi_k\rangle|^2 = \frac{d\delta_{jk} +1}{d+1},\quad j,k=1, 2, \ldots, d^2.
\end{equation}
Other prominent examples of  2-designs include
MUB~\cite{KlapR05M}, which are the focus of this paper.

A set of $K$ unitary operators $\{U_j\}$ is a \emph{unitary 2-design} \cite{Dank05the, DankCEL09, GrosAE07} if it satisfies
\begin{equation}\label{eq:U2design}
\frac{1}{K} \sum _j (U_j\otimes U_j)A(U_j\otimes U_j)^\dag =\int \rmd U (U\otimes U)A(U\otimes U)^\dag
\end{equation}
for any operator $A$ acting on the bipartite Hilbert space, where the integral is taken over the whole unitary group with respect to the normalized Haar measure. Unitary 2-designs have found applications in quantum process estimation \cite{Dank05the, DankCEL09} and  quantum cryptography \cite{Chau05}. When $A$ is the tensor product of a pure state with itself, the right hand side of \eref{eq:U2design} is proportional to the projector onto the symmetric subspace. It follows that any orbit of pure states  of a unitary 2-design (that forms a group) is a (projective) 2-design. In particular, this conclusion applies to the orbits of the restricted Clifford group,  which is a unitary 2-design according to Chau \cite{Chau05} (see also \rcite{GrosAE07}). This means  that  infinitely many inequivalent  2-designs (each with finite number of elements) can be constructed as orbits of the restricted Clifford group. Besides the Hesse SIC and canonical MUB, another notable  2-design of this type is the one composed of MUB balanced states \cite{Appl09P,AmbuSSW14, ApplBD14}.

 Among  Clifford covariant 2-designs, the one with the least number of elements is usually the most interesting from both theoretical and practical perspectives.
\Thref{thm:MinOrbit} shows that, except in dimension 3, any Clifford covariant 2-design has at least $d(d+1)$ elements, and the lower bound is saturated if and only if the 2-design corresponds to the canonical MUB; in other words, the canonical MUB is the unique minimal 2-design covariant with respect to the Clifford group.
In addition, this theorem  implies that the Hesse SIC is the unique SIC that is covariant with respect to the restricted Clifford group, which complements a similar conclusion derived in \rcite{Zhu15P}. Although \Thref{thm:MinOrbit} is already satisfactory as it stands, we have a feeling that even stronger conclusions may be derived.
\begin{conjecture}\label{con:minimal2design}
Any 2-design in dimension $d$ with no more than $d(d+1)$ elements is either a SIC or MUB.
\end{conjecture}
\begin{conjecture}\label{con:sym2design}
When $d=q=p^n$ is a prime power, the symmetry group of any such 2-design has order at most $nq^3(q^2-1)$, and the upper bound is saturated if and only if the 2-design corresponds to the Hesse SIC in dimension~3 or the canonical MUB.
\end{conjecture}
The first conjecture holds in the special cases in which the pairwise inner products are restricted to the single value of $1/(d+1)$ or two values of 0 and $1/d$ \cite{KlapR05M}. In general, 2-designs with the number of elements close to the minimum $d^2$ are quite rare, and
pairwise inner products in such a 2-design usually take on special values. Therefore, it is reasonable to expect that this conjecture holds in general. The second conjecture is more plausible  given that discrete structures with high symmetry in finite-state quantum mechanics, such as Wootters discrete Wigner functions \cite{Woot87, Gros06, Zhu15P} and SICs \cite{Zaun11, ReneBSC04, Appl05, Zhu10, Zhu14S},  are often connected with the Clifford group or the restricted Clifford group. With the group theoretic structure at our disposal, it is even likely that this  conjecture may be proved within a couple of years after we know more about MUB,  2-designs, and Clifford groups.

Before proving \thref{thm:MinOrbit}, we need to introduce three technical lemmas. The first one is reproduced from Lemma 7.2 in the author's thesis \cite{Zhu12the} (see also \rscite{Zhu14S,Zhu15Sh,Zhu15P}).
\begin{lemma}\label{lem:orbits}
Suppose $G$ is a subgroup of the symmetry group of a SIC. Then
the number of orbits  of
$G$ on the SIC is equal to the sum of squared multiplicities of all the
inequivalent irreducible components of $G$. In particular, $G$ acts transitively on the SIC if and only if it is irreducible.
\end{lemma}

\begin{lemma}\label{lem:MinOrbit}
Any orbit of the restricted Clifford group in dimension $q=p^n$ has more than $q^2$ elements except when $q=3$ and the orbit corresponds to the Hesse SIC.
\end{lemma}
\begin{proof}
The Hesse SIC  in dimension 3    is  generated by the HW group   from the fiducial ket
\begin{equation}\label{eq:GodSIC}
|\psi_3\rangle \rep\frac{1}{\sqrt{2}}
(0, 1, -1)^\rmT.
\end{equation}
It is indeed covariant with respect to the (restricted)  Clifford group \cite{Appl05, Zhu10, Zhu14S}.

The restricted Clifford group is a unitary 2-design \cite{Chau05,GrosAE07}, so any orbit constitutes a  2-design. Any 2-design has at least $q^2$ elements, and the lower bound is saturated if and only if the 2-design corresponds to a SIC \cite{Zaun11, ReneBSC04, Scot06, ApplFZ15G}. Suppose the orbit of the restricted Clifford group has $q^2$ elements, then it is a SIC and  is covariant with respect to the HW group according to \lref{lem:orbits}. In any even prime power dimension other than 2 and 8, no SIC is covariant with respect to the HW group (the one defined in \eref{eq:HWmul}) according to \rcite{GodsR09}, which leads to a contradiction. In general, the stabilizer of each state on the orbit is isomorphic to $\SL{2}{q}$ and acts transitively on nontrivial displacement operators. Consequently, the restricted Clifford group acts doubly transitively on the states in the SIC; that is,  the SIC is super-symmetric as defined by the author in \rscite{Zhu14S, Zhu15P}. However, the SIC in dimension~2, the Hesse SIC in dimension~3, and the set of Hoggar lines in dimension~8 are the only three super-symmetric SICs; in addition, only the Hesse SIC is covariant with respect to the restricted Clifford group \cite{Zhu14S, Zhu15P}\footnote{Although the general results in \rscite{Zhu14S, Zhu15P} were derived by virtue of the classification of finite simple groups (CFSG),  the conclusion needed  here does not rely on the CFSG, since we are only concerned with super-symmetric SICs that are covariant with respect to the restricted Clifford group in a prime power dimension.}.
\end{proof}

\begin{lemma}\label{lem:SylowNormalizer}
Suppose  $\overline{S}$ is a subgroup of the restricted  Clifford group $\pcr$ of index $q(q+1)$ whose intersection  with the HW group $\overline{S}\cap\phw$ has order $q$. Then any subgroup $T$ of $D$ with collineation group $\overline{T}=\overline{S}\cap\phw$ is a maximal abelian subgroup (up to a phase factor) of the HW group corresponding to a ray in $\bbF_q^2$, and any fixed point of  $\overline{S}$ belongs to the canonical MUB.
\end{lemma}
\begin{proof}
Note that  $R:=\overline{S}/\overline{T}$ can be identified with an index-($q+1$) subgroup of $\SL{2}{q}$.
Let $A$ be the center of $\SL{2}{q}$, which is trivial for even $q$ and has order 2 otherwise \cite{Dick58book}.   Then $RA/A$ can be identified as a subgroup of  the projective special linear group $\PSL{2}{q}:=\SL{2}{q}/A$ of index either $q+1$ or $(q+1)/2$ (for odd $q$).
According to Chapter XII in \rcite{Dick58book}, $\PSL{2}{q}$ has no subgroup of index $(q+1)/2$, and all subgroups of index $q+1$ are normalizers of Sylow $p$-subgroups, which are conjugate to each other. It follows that all  index-($q+1$) subgroups of $\SL{2}{q}$
are normalizers of Sylow $p$-subgroups and are conjugate to each other. Without loss of generality, we may assume that  $R$
is  composed of the following elements,
\begin{equation}
\begin{pmatrix}
\alpha &0\\ \gamma &\alpha^{-1}
\end{pmatrix},\quad \alpha \in \bbF_{q}^*,\quad \gamma \in \bbF_{q}.
\end{equation}
Note that $R$ has a unique nontrivial invariant subspace in $\bbF_q^2$, which happen to be the ray  composed of vectors of the form $(0, \xi)^T$ with $\xi\in \bbF_q$. The same conclusion still holds even if $R$ is taken as a linear group on $\bbF_{p}^{2n}$. Since $\overline{T}$ is a nontrivial normal subgroup of $\overline{S}$, it corresponds to a nontrivial invariant subspace of $R$ over $\bbF_{p}^{2n}$ and is thus uniquely determined by the ray. Therefore, $T$ is a maximal abelian subgroup of the HW group corresponding to a ray in $\bbF_q^2$, and any fixed point of  $\overline{S}$ belongs to the canonical MUB.
\end{proof}

Now we are ready to prove our main result.
\begin{proof}[Proof of \thref{thm:MinOrbit}]
The conclusion follows from \lref{lem:MinOrbit} when $q=3$. Suppose $O$ is a characteristic orbit of the restricted Clifford group other than the Hesse SIC. Let $S$ be the stabilizer of a given element in $O$, $T=S\cap\hw$, and $R=\overline{S}/\overline{T}$. Then $T$ is an abelian subgroup of $D$; otherwise, each irreducible component of $T$ has a degree of at least $p$, so that $T$ cannot stabilize any pure state.
As a consequence, the order of $\overline{T}$ divides $q$, and its index in $\phw$ is divisible by $q$. Since $|O|$ is equal to the index of $\overline{S}$ in the restricted Clifford group $\pcr$, it follows that $|O|$ is divisible by $q$, which implies that $|O|\geq q(q+1)$ given that $|O|>q^2$ according to \lref{lem:MinOrbit}. If the lower bound is saturated,
then $|\overline{S}|$ has index $q(q+1)$ in  $\pcr$, $|\overline{T}|=q$, and $R$ has index $q+1$ in $\SL{2}{q}$, that is,
$|R|=q(q-1)$ (recall that $\SL{2}{q}$ has order $q(q^2-1)$). Consequently, $T$ is a maximal abelian subgroup of the HW group, and $O$ is composed of  stabilizer states. Observing that the HW group acts transitively on the states in each stabilizer basis, we conclude that $O$ is composed of $q+1$ stabilizer bases. In addition, the $q+1$ bases are mutually unbiased because any orbit of the Clifford group is a 2-design, while $q+1$ bases form a 2-design if and only if they are mutually unbiased.  Moreover, any fixed  point of  $\overline{S}$ belongs to the canonical MUB according to \lref{lem:SylowNormalizer}, so the orbit $O$ defines the canonical MUB.
\end{proof}

Before concluding this paper, we mention a surprising   application of our  study  to frame theory \cite{CasaK13}.
A frame in dimension $d$ is a set of pure states $B=\{|\psi_j\rangle\}$ that  spans the whole Hilbert space; the frame is tight if $\sum_j|\psi_j\rangle\langle \psi_j|$ is proportional to the identity, in which case the frame  defines a POVM after scaling. The symmetry group $\overline{H}$ of the frame is the group composed of all unitary transformations that leave the frame invariant (as in the case of a MUB).  Note that all pure states (including those not contained in $B$) form disjoint orbits under the action of $\overline{H}$. Assuming that $\overline{H}$ is irreducible, the frame $B$ is necessarily tight. The frame $B$ is \emph{highly symmetric} if its symmetry group acts transitively on $B$  and the stabilizer of each state in $B$ is not properly contained in the stabilizer of any other pure state.  Highly symmetric frames are interesting in the study of  signal processing; the POVMs constructed from them, called highly symmetric POVMs,  are  useful in the study of entropic uncertainty relations, informational power, etc.  \cite{BrooW13, CasaK13, SlomS14}.

Given an irreducible subgroup $\overline{G}$ of the (projective) unitary group, if it has a unique characteristic orbit, then the frame corresponding to the orbit is automatically highly symmetric, though the converse does not hold in general.
In addition, a frame building on an orbit of $\overline{G}$ is highly symmetric if all the fixed points of the stabilizer of each state on the orbit also belong to the orbit. In view of this observation, \thref{thm:MinOrbit} implies that the Hesse SIC and canonical MUB are  highly symmetric, thereby   providing an infinite family of highly symmetric frames and highly symmetric POVMs. It should be noted that the property of being highly symmetric is not unique to the canonical MUB. Actually, any group covariant frame composed of stabilizer bases is also highly symmetric. In particular, this conclusion applies to the frame composed of   all stabilizer states. The implications of this observation deserve further study.

In summary, except in dimension~3, the canonical MUB is the unique orbit of the restricted Clifford group with the minimal cardinality; it is also the unique minimal 2-design covariant with respect to the restricted Clifford group. Remarkably, this MUB can be characterized in such a simple way that does not involve any concept related to bases or unbiasedness.  In addition,  the Hesse SIC in dimension 3 is the unique SIC that is covariant with respect to the restricted Clifford group. SICs and MUB are often dubbed  as sets of states that have no right to exist. Our study provides a new perspective for understanding this problem, which may stimulate further progress along this direction. The ideas introduced here are also useful to studying other discrete symmetric structures, such as discrete Wigner functions. As a by-product, our study provides an infinite family of highly symmetric frames and POVMs, which are of independent interest.

Note added: upon completion of this paper, we noticed a related work \cite{CarmST15}.

\section*{Acknowledgments}
Thanks to Mark Howard and Markus Grassl for
pointing out \rcite{CarmST15}.
 This work is supported in part by Perimeter Institute for Theoretical Physics. Research at Perimeter Institute is supported by the Government of Canada through Industry Canada and by the Province of Ontario through the Ministry of Research and Innovation.

\bibliographystyle{apsrev4-1}

\bibliography{all_references}

\begin{thebibliography}{45}%
\makeatletter
\providecommand \@ifxundefined [1]{%
 \@ifx{#1\undefined}
}%
\providecommand \@ifnum [1]{%
 \ifnum #1\expandafter \@firstoftwo
 \else \expandafter \@secondoftwo
 \fi
}%
\providecommand \@ifx [1]{%
 \ifx #1\expandafter \@firstoftwo
 \else \expandafter \@secondoftwo
 \fi
}%
\providecommand \natexlab [1]{#1}%
\providecommand \enquote  [1]{``#1''}%
\providecommand \bibnamefont  [1]{#1}%
\providecommand \bibfnamefont [1]{#1}%
\providecommand \citenamefont [1]{#1}%
\providecommand \href@noop [0]{\@secondoftwo}%
\providecommand \href [0]{\begingroup \@sanitize@url \@href}%
\providecommand \@href[1]{\@@startlink{#1}\@@href}%
\providecommand \@@href[1]{\endgroup#1\@@endlink}%
\providecommand \@sanitize@url [0]{\catcode `\\12\catcode `\$12\catcode
  `\&12\catcode `\#12\catcode `\^12\catcode `\_12\catcode `\%12\relax}%
\providecommand \@@startlink[1]{}%
\providecommand \@@endlink[0]{}%
\providecommand \url  [0]{\begingroup\@sanitize@url \@url }%
\providecommand \@url [1]{\endgroup\@href {#1}{\urlprefix }}%
\providecommand \urlprefix  [0]{URL }%
\providecommand \Eprint [0]{\href }%
\providecommand \doibase [0]{http://dx.doi.org/}%
\providecommand \selectlanguage [0]{\@gobble}%
\providecommand \bibinfo  [0]{\@secondoftwo}%
\providecommand \bibfield  [0]{\@secondoftwo}%
\providecommand \translation [1]{[#1]}%
\providecommand \BibitemOpen [0]{}%
\providecommand \bibitemStop [0]{}%
\providecommand \bibitemNoStop [0]{.\EOS\space}%
\providecommand \EOS [0]{\spacefactor3000\relax}%
\providecommand \BibitemShut  [1]{\csname bibitem#1\endcsname}%
\let\auto@bib@innerbib\@empty
\bibitem [{\citenamefont {Durt}\ \emph {et~al.}(2010)\citenamefont {Durt},
  \citenamefont {Englert}, \citenamefont {Bengtsson},\ and\ \citenamefont
  {{\.{Z}}yczkowski}}]{DurtEBZ10}%
  \BibitemOpen
  \bibfield  {author} {\bibinfo {author} {\bibfnamefont {T.}~\bibnamefont
  {Durt}}, \bibinfo {author} {\bibfnamefont {B.-G.}\ \bibnamefont {Englert}},
  \bibinfo {author} {\bibfnamefont {I.}~\bibnamefont {Bengtsson}}, \ and\
  \bibinfo {author} {\bibfnamefont {K.}~\bibnamefont {{\.{Z}}yczkowski}},\
  }\href@noop {} {\bibfield  {journal} {\bibinfo  {journal} {Int. J. Quant.
  Inf.}\ }\textbf {\bibinfo {volume} {8}},\ \bibinfo {pages} {535} (\bibinfo
  {year} {2010})}\BibitemShut {NoStop}%
\bibitem [{\citenamefont {Wootters}\ and\ \citenamefont
  {Fields}(1989)}]{WootF89}%
  \BibitemOpen
  \bibfield  {author} {\bibinfo {author} {\bibfnamefont {W.~K.}\ \bibnamefont
  {Wootters}}\ and\ \bibinfo {author} {\bibfnamefont {B.~D.}\ \bibnamefont
  {Fields}},\ }\href@noop {} {\bibfield  {journal} {\bibinfo  {journal} {Ann.
  Phys.}\ }\textbf {\bibinfo {volume} {191}},\ \bibinfo {pages} {363} (\bibinfo
  {year} {1989})}\BibitemShut {NoStop}%
\bibitem [{\citenamefont {Ivanovi\'c}(1981)}]{Ivan81}%
  \BibitemOpen
  \bibfield  {author} {\bibinfo {author} {\bibfnamefont {I.~D.}\ \bibnamefont
  {Ivanovi\'c}},\ }\href@noop {} {\bibfield  {journal} {\bibinfo  {journal} {J.
  Phys. A: Math. Gen.}\ }\textbf {\bibinfo {volume} {14}},\ \bibinfo {pages}
  {3241} (\bibinfo {year} {1981})}\BibitemShut {NoStop}%
\bibitem [{\citenamefont {Alltop}(1980)}]{Allt80}%
  \BibitemOpen
  \bibfield  {author} {\bibinfo {author} {\bibfnamefont {W.}~\bibnamefont
  {Alltop}},\ }\href@noop {} {\bibfield  {journal} {\bibinfo  {journal} {IEEE
  Trans. Inf. Theory}\ }\textbf {\bibinfo {volume} {26}},\ \bibinfo {pages}
  {350} (\bibinfo {year} {1980})}\BibitemShut {NoStop}%
\bibitem [{\citenamefont {Popa}(1983)}]{Popa83}%
  \BibitemOpen
  \bibfield  {author} {\bibinfo {author} {\bibfnamefont {S.}~\bibnamefont
  {Popa}},\ }\href@noop {} {\bibfield  {journal} {\bibinfo  {journal} {J.
  Operator Theory}\ }\textbf {\bibinfo {volume} {9}},\ \bibinfo {pages} {253}
  (\bibinfo {year} {1983})}\BibitemShut {NoStop}%
\bibitem [{\citenamefont {Gottesman}(1997)}]{Gott97the}%
  \BibitemOpen
  \bibfield  {author} {\bibinfo {author} {\bibfnamefont {D.}~\bibnamefont
  {Gottesman}},\ }\emph {\bibinfo {title} {Stabilizer Codes and Quantum Error
  Correction}},\ \href@noop {} {Ph.D. thesis},\ \bibinfo  {school} {California
  Institute of Technology} (\bibinfo {year} {1997}),\ \bibinfo {note}
  {available at \url{http://arxiv.org/abs/quant-ph/9705052}}\BibitemShut
  {NoStop}%
\bibitem [{\citenamefont {Lawrence}\ \emph {et~al.}(2002)\citenamefont
  {Lawrence}, \citenamefont {Brukner},\ and\ \citenamefont
  {Zeilinger}}]{LawrBZ02}%
  \BibitemOpen
  \bibfield  {author} {\bibinfo {author} {\bibfnamefont {J.}~\bibnamefont
  {Lawrence}}, \bibinfo {author} {\bibfnamefont {{\v{C}}.}~\bibnamefont
  {Brukner}}, \ and\ \bibinfo {author} {\bibfnamefont {A.}~\bibnamefont
  {Zeilinger}},\ }\href@noop {} {\bibfield  {journal} {\bibinfo  {journal}
  {Phys. Rev. A}\ }\textbf {\bibinfo {volume} {65}},\ \bibinfo {pages} {032320}
  (\bibinfo {year} {2002})}\BibitemShut {NoStop}%
\bibitem [{\citenamefont {Bandyopadhyay}\ \emph {et~al.}(2002)\citenamefont
  {Bandyopadhyay}, \citenamefont {Boykin}, \citenamefont {Roychowdhury},\ and\
  \citenamefont {Vatan}}]{BandBRV02}%
  \BibitemOpen
  \bibfield  {author} {\bibinfo {author} {\bibfnamefont {S.}~\bibnamefont
  {Bandyopadhyay}}, \bibinfo {author} {\bibfnamefont {P.~O.}\ \bibnamefont
  {Boykin}}, \bibinfo {author} {\bibfnamefont {V.}~\bibnamefont
  {Roychowdhury}}, \ and\ \bibinfo {author} {\bibfnamefont {F.}~\bibnamefont
  {Vatan}},\ }\href@noop {} {\bibfield  {journal} {\bibinfo  {journal}
  {Algorithmica}\ }\textbf {\bibinfo {volume} {34}},\ \bibinfo {pages} {512}
  (\bibinfo {year} {2002})}\BibitemShut {NoStop}%
\bibitem [{\citenamefont {Godsil}\ and\ \citenamefont {Roy}(2009)}]{GodsR09}%
  \BibitemOpen
  \bibfield  {author} {\bibinfo {author} {\bibfnamefont {C.}~\bibnamefont
  {Godsil}}\ and\ \bibinfo {author} {\bibfnamefont {A.}~\bibnamefont {Roy}},\
  }\href@noop {} {\bibfield  {journal} {\bibinfo  {journal} {Eur. J.
  Combinator.}\ }\textbf {\bibinfo {volume} {30}},\ \bibinfo {pages} {246}
  (\bibinfo {year} {2009})}\BibitemShut {NoStop}%
\bibitem [{\citenamefont {Kantor}(2012)}]{Kant12}%
  \BibitemOpen
  \bibfield  {author} {\bibinfo {author} {\bibfnamefont {W.~M.}\ \bibnamefont
  {Kantor}},\ }\href@noop {} {\bibfield  {journal} {\bibinfo  {journal} {J.
  Math. Phys.}\ }\textbf {\bibinfo {volume} {53}},\ \bibinfo {eid} {032204}
  (\bibinfo {year} {2012})}\BibitemShut {NoStop}%
\bibitem [{\citenamefont {Zhu}(2015{\natexlab{a}})}]{Zhu15Sh}%
  \BibitemOpen
  \bibfield  {author} {\bibinfo {author} {\bibfnamefont {H.}~\bibnamefont
  {Zhu}},\ }\href@noop {} {\enquote {\bibinfo {title} {{Sharply covariant
  mutually unbiased bases}},}\ } (\bibinfo {year} {2015}{\natexlab{a}}),\
  \Eprint {http://arxiv.org/abs/1503.00003} {arXiv:1503.00003} \BibitemShut
  {NoStop}%
\bibitem [{\citenamefont {Cormick}\ \emph {et~al.}(2006)\citenamefont
  {Cormick}, \citenamefont {Galv\~ao}, \citenamefont {Gottesman}, \citenamefont
  {Paz},\ and\ \citenamefont {Pittenger}}]{CormGGP06}%
  \BibitemOpen
  \bibfield  {author} {\bibinfo {author} {\bibfnamefont {C.}~\bibnamefont
  {Cormick}}, \bibinfo {author} {\bibfnamefont {E.~F.}\ \bibnamefont
  {Galv\~ao}}, \bibinfo {author} {\bibfnamefont {D.}~\bibnamefont {Gottesman}},
  \bibinfo {author} {\bibfnamefont {J.~P.}\ \bibnamefont {Paz}}, \ and\
  \bibinfo {author} {\bibfnamefont {A.~O.}\ \bibnamefont {Pittenger}},\
  }\href@noop {} {\bibfield  {journal} {\bibinfo  {journal} {Phys. Rev. A}\
  }\textbf {\bibinfo {volume} {73}},\ \bibinfo {pages} {012301} (\bibinfo
  {year} {2006})}\BibitemShut {NoStop}%
\bibitem [{\citenamefont {Appleby}\ \emph
  {et~al.}(2014{\natexlab{a}})\citenamefont {Appleby}, \citenamefont {Dang},\
  and\ \citenamefont {Fuchs}}]{ApplDF14}%
  \BibitemOpen
  \bibfield  {author} {\bibinfo {author} {\bibfnamefont {D.~M.}\ \bibnamefont
  {Appleby}}, \bibinfo {author} {\bibfnamefont {H.~B.}\ \bibnamefont {Dang}}, \
  and\ \bibinfo {author} {\bibfnamefont {C.~A.}\ \bibnamefont {Fuchs}},\
  }\href@noop {} {\bibfield  {journal} {\bibinfo  {journal} {Entropy}\ }\textbf
  {\bibinfo {volume} {16}},\ \bibinfo {pages} {1484} (\bibinfo {year}
  {2014}{\natexlab{a}})}\BibitemShut {NoStop}%
\bibitem [{\citenamefont {Appleby}(2009)}]{Appl09P}%
  \BibitemOpen
  \bibfield  {author} {\bibinfo {author} {\bibfnamefont {D.~M.}\ \bibnamefont
  {Appleby}},\ }\href@noop {} {\enquote {\bibinfo {title} {{Properties of the
  extended Clifford group with applications to SIC-POVMs and MUBs}},}\ }
  (\bibinfo {year} {2009}),\ \Eprint {http://arxiv.org/abs/0909.5233}
  {arXiv:0909.5233} \BibitemShut {NoStop}%
\bibitem [{\citenamefont {Zauner}(2011)}]{Zaun11}%
  \BibitemOpen
  \bibfield  {author} {\bibinfo {author} {\bibfnamefont {G.}~\bibnamefont
  {Zauner}},\ }\href@noop {} {\bibfield  {journal} {\bibinfo  {journal} {Int.
  J. Quant. Inf.}\ }\textbf {\bibinfo {volume} {9}},\ \bibinfo {pages} {445}
  (\bibinfo {year} {2011})}\BibitemShut {NoStop}%
\bibitem [{\citenamefont {Appleby}(2005)}]{Appl05}%
  \BibitemOpen
  \bibfield  {author} {\bibinfo {author} {\bibfnamefont {D.~M.}\ \bibnamefont
  {Appleby}},\ }\href@noop {} {\bibfield  {journal} {\bibinfo  {journal} {J.
  Math. Phys.}\ }\textbf {\bibinfo {volume} {46}},\ \bibinfo {pages} {052107}
  (\bibinfo {year} {2005})}\BibitemShut {NoStop}%
\bibitem [{\citenamefont {Zhu}(2010)}]{Zhu10}%
  \BibitemOpen
  \bibfield  {author} {\bibinfo {author} {\bibfnamefont {H.}~\bibnamefont
  {Zhu}},\ }\href@noop {} {\bibfield  {journal} {\bibinfo  {journal} {J. Phys.
  A: Math. Theor.}\ }\textbf {\bibinfo {volume} {43}},\ \bibinfo {pages}
  {305305} (\bibinfo {year} {2010})}\BibitemShut {NoStop}%
\bibitem [{\citenamefont {Zhu}(2014{\natexlab{a}})}]{Zhu14S}%
  \BibitemOpen
  \bibfield  {author} {\bibinfo {author} {\bibfnamefont {H.}~\bibnamefont
  {Zhu}},\ }\href@noop {} {\enquote {\bibinfo {title} {Super-symmetric
  informationally complete measurements},}\ } (\bibinfo {year}
  {2014}{\natexlab{a}}),\ \Eprint {http://arxiv.org/abs/1412.1099}
  {arXiv:1412.1099} \BibitemShut {NoStop}%
\bibitem [{\citenamefont {Zhu}(2015{\natexlab{b}})}]{Zhu15P}%
  \BibitemOpen
  \bibfield  {author} {\bibinfo {author} {\bibfnamefont {H.}~\bibnamefont
  {Zhu}},\ }\href@noop {} {\enquote {\bibinfo {title} {{Permutation symmetry
  determines the discrete Wigner function}},}\ } (\bibinfo {year}
  {2015}{\natexlab{b}}),\ \Eprint {http://arxiv.org/abs/1504.03773}
  {arXiv:1504.03773} \BibitemShut {NoStop}%
\bibitem [{\citenamefont {Broome}\ and\ \citenamefont
  {Waldron}(2013)}]{BrooW13}%
  \BibitemOpen
  \bibfield  {author} {\bibinfo {author} {\bibfnamefont {H.}~\bibnamefont
  {Broome}}\ and\ \bibinfo {author} {\bibfnamefont {S.}~\bibnamefont
  {Waldron}},\ }\href@noop {} {\bibfield  {journal} {\bibinfo  {journal}
  {Linear Algebra and its Applications}\ }\textbf {\bibinfo {volume} {439}},\
  \bibinfo {pages} {4135 } (\bibinfo {year} {2013})}\BibitemShut {NoStop}%
\bibitem [{\citenamefont {Casazza}\ and\ \citenamefont
  {Kutyniok}(2013)}]{CasaK13}%
  \BibitemOpen
  \bibinfo {editor} {\bibfnamefont {P.~G.}\ \bibnamefont {Casazza}}\ and\
  \bibinfo {editor} {\bibfnamefont {G.}~\bibnamefont {Kutyniok}},\ eds.,\
  \href@noop {} {\emph {\bibinfo {title} {Finite Frames}}},\ Applied and
  Numerical Harmonic Analysis\ (\bibinfo  {publisher} {Birkh\"auser, New
  York},\ \bibinfo {year} {2013})\BibitemShut {NoStop}%
\bibitem [{\citenamefont {S{\l}omczy\'{n}ski}\ and\ \citenamefont
  {Szymusiak}(2014)}]{SlomS14}%
  \BibitemOpen
  \bibfield  {author} {\bibinfo {author} {\bibfnamefont {W.}~\bibnamefont
  {S{\l}omczy\'{n}ski}}\ and\ \bibinfo {author} {\bibfnamefont
  {A.}~\bibnamefont {Szymusiak}},\ }\href {http://arxiv.org/abs/1402.0375}
  {\enquote {\bibinfo {title} {{Highly symmetric POVMs and their informational
  power}},}\ } (\bibinfo {year} {2014}),\ \Eprint
  {http://arxiv.org/abs/1402.0375} {arXiv:1402.0375} \BibitemShut {NoStop}%
\bibitem [{\citenamefont {Aschbacher}\ \emph {et~al.}(2007)\citenamefont
  {Aschbacher}, \citenamefont {Childs},\ and\ \citenamefont
  {Wocjan}}]{AschCW07}%
  \BibitemOpen
  \bibfield  {author} {\bibinfo {author} {\bibfnamefont {M.}~\bibnamefont
  {Aschbacher}}, \bibinfo {author} {\bibfnamefont {A.~M.}\ \bibnamefont
  {Childs}}, \ and\ \bibinfo {author} {\bibfnamefont {P.}~\bibnamefont
  {Wocjan}},\ }\href@noop {} {\bibfield  {journal} {\bibinfo  {journal} {J.
  Algebr. Comb.}\ }\textbf {\bibinfo {volume} {25}},\ \bibinfo {pages} {111}
  (\bibinfo {year} {2007})}\BibitemShut {NoStop}%
\bibitem [{\citenamefont {Gross}(2006)}]{Gros06}%
  \BibitemOpen
  \bibfield  {author} {\bibinfo {author} {\bibfnamefont {D.}~\bibnamefont
  {Gross}},\ }\href@noop {} {\bibfield  {journal} {\bibinfo  {journal} {J.
  Math. Phys.}\ }\textbf {\bibinfo {volume} {47}},\ \bibinfo {eid} {122107}
  (\bibinfo {year} {2006})}\BibitemShut {NoStop}%
\bibitem [{\citenamefont {Bolt}\ \emph
  {et~al.}(1961{\natexlab{a}})\citenamefont {Bolt}, \citenamefont {Room},\ and\
  \citenamefont {Wall}}]{BoltRW61I}%
  \BibitemOpen
  \bibfield  {author} {\bibinfo {author} {\bibfnamefont {B.}~\bibnamefont
  {Bolt}}, \bibinfo {author} {\bibfnamefont {T.~G.}\ \bibnamefont {Room}}, \
  and\ \bibinfo {author} {\bibfnamefont {G.~E.}\ \bibnamefont {Wall}},\
  }\href@noop {} {\bibfield  {journal} {\bibinfo  {journal} {J. Austral. Math.
  Soc.}\ }\textbf {\bibinfo {volume} {2}},\ \bibinfo {pages} {60} (\bibinfo
  {year} {1961}{\natexlab{a}})}\BibitemShut {NoStop}%
\bibitem [{\citenamefont {Bolt}\ \emph
  {et~al.}(1961{\natexlab{b}})\citenamefont {Bolt}, \citenamefont {Room},\ and\
  \citenamefont {Wall}}]{BoltRW61II}%
  \BibitemOpen
  \bibfield  {author} {\bibinfo {author} {\bibfnamefont {B.}~\bibnamefont
  {Bolt}}, \bibinfo {author} {\bibfnamefont {T.~G.}\ \bibnamefont {Room}}, \
  and\ \bibinfo {author} {\bibfnamefont {G.~E.}\ \bibnamefont {Wall}},\
  }\href@noop {} {\bibfield  {journal} {\bibinfo  {journal} {J. Austral. Math.
  Soc.}\ }\textbf {\bibinfo {volume} {2}},\ \bibinfo {pages} {80} (\bibinfo
  {year} {1961}{\natexlab{b}})}\BibitemShut {NoStop}%
\bibitem [{\citenamefont {Taylor}(2006)}]{Tayl06}%
  \BibitemOpen
  \bibfield  {author} {\bibinfo {author} {\bibfnamefont {D.~E.}\ \bibnamefont
  {Taylor}},\ }\href@noop {} {\enquote {\bibinfo {title} {Pairs of generators
  for matrix groups.~{I}},}\ } (\bibinfo {year} {2006}),\ \bibinfo {note}
  {available at \url{www.maths.usyd.edu.au/u/don/papers/genAC.pdf}}\BibitemShut
  {NoStop}%
\bibitem [{\citenamefont {Dehaene}\ and\ \citenamefont
  {De~Moor}(2003)}]{DehaM03}%
  \BibitemOpen
  \bibfield  {author} {\bibinfo {author} {\bibfnamefont {J.}~\bibnamefont
  {Dehaene}}\ and\ \bibinfo {author} {\bibfnamefont {B.}~\bibnamefont
  {De~Moor}},\ }\href@noop {} {\bibfield  {journal} {\bibinfo  {journal} {Phys.
  Rev. A}\ }\textbf {\bibinfo {volume} {68}},\ \bibinfo {pages} {042318}
  (\bibinfo {year} {2003})}\BibitemShut {NoStop}%
\bibitem [{\citenamefont {Kleidman}\ and\ \citenamefont
  {Liebeck}(1990)}]{KleiL90}%
  \BibitemOpen
  \bibfield  {author} {\bibinfo {author} {\bibfnamefont {P.~B.}\ \bibnamefont
  {Kleidman}}\ and\ \bibinfo {author} {\bibfnamefont {M.~W.}\ \bibnamefont
  {Liebeck}},\ }\href@noop {} {\emph {\bibinfo {title} {The Subgroup Structure
  of the Finite Classical Groups}}},\ \bibinfo {series} {London Mathematical
  Society Lecture Note Series}, Vol.\ \bibinfo {volume} {129}\ (\bibinfo
  {publisher} {Cambridge University Press},\ \bibinfo {year}
  {1990})\BibitemShut {NoStop}%
\bibitem [{\citenamefont {Renes}\ \emph {et~al.}(2004)\citenamefont {Renes},
  \citenamefont {Blume-Kohout}, \citenamefont {Scott},\ and\ \citenamefont
  {Caves}}]{ReneBSC04}%
  \BibitemOpen
  \bibfield  {author} {\bibinfo {author} {\bibfnamefont {J.~M.}\ \bibnamefont
  {Renes}}, \bibinfo {author} {\bibfnamefont {R.}~\bibnamefont {Blume-Kohout}},
  \bibinfo {author} {\bibfnamefont {A.~J.}\ \bibnamefont {Scott}}, \ and\
  \bibinfo {author} {\bibfnamefont {C.~M.}\ \bibnamefont {Caves}},\ }\href@noop
  {} {\bibfield  {journal} {\bibinfo  {journal} {J. Math. Phys.}\ }\textbf
  {\bibinfo {volume} {45}},\ \bibinfo {pages} {2171} (\bibinfo {year}
  {2004})},\ \bibinfo {note} {supplementary information including the fiducial
  kets available at \url{http://www.cquic.org/papers/reports/}}\BibitemShut
  {NoStop}%
\bibitem [{\citenamefont {Scott}(2006)}]{Scot06}%
  \BibitemOpen
  \bibfield  {author} {\bibinfo {author} {\bibfnamefont {A.~J.}\ \bibnamefont
  {Scott}},\ }\href@noop {} {\bibfield  {journal} {\bibinfo  {journal} {J.
  Phys. A: Math. Gen.}\ }\textbf {\bibinfo {volume} {39}},\ \bibinfo {pages}
  {13507} (\bibinfo {year} {2006})}\BibitemShut {NoStop}%
\bibitem [{\citenamefont {Appleby}\ \emph {et~al.}(2015)\citenamefont
  {Appleby}, \citenamefont {Fuchs},\ and\ \citenamefont {Zhu}}]{ApplFZ15G}%
  \BibitemOpen
  \bibfield  {author} {\bibinfo {author} {\bibfnamefont {D.~M.}\ \bibnamefont
  {Appleby}}, \bibinfo {author} {\bibfnamefont {C.~A.}\ \bibnamefont {Fuchs}},
  \ and\ \bibinfo {author} {\bibfnamefont {H.}~\bibnamefont {Zhu}},\
  }\href@noop {} {\bibfield  {journal} {\bibinfo  {journal} {Quantum Inf.
  Comput.}\ }\textbf {\bibinfo {volume} {15}},\ \bibinfo {pages} {61} (\bibinfo
  {year} {2015})}\BibitemShut {NoStop}%
\bibitem [{\citenamefont {Zhu}\ and\ \citenamefont {Englert}(2011)}]{ZhuE11}%
  \BibitemOpen
  \bibfield  {author} {\bibinfo {author} {\bibfnamefont {H.}~\bibnamefont
  {Zhu}}\ and\ \bibinfo {author} {\bibfnamefont {B.-G.}\ \bibnamefont
  {Englert}},\ }\href@noop {} {\bibfield  {journal} {\bibinfo  {journal} {Phys.
  Rev. A}\ }\textbf {\bibinfo {volume} {84}},\ \bibinfo {pages} {022327}
  (\bibinfo {year} {2011})}\BibitemShut {NoStop}%
\bibitem [{\citenamefont {Zhu}(2012)}]{Zhu12the}%
  \BibitemOpen
  \bibfield  {author} {\bibinfo {author} {\bibfnamefont {H.}~\bibnamefont
  {Zhu}},\ }\emph {\bibinfo {title} {Quantum State Estimation and Symmetric
  Informationally Complete {POM}s}},\ \href@noop {} {Ph.D. thesis},\ \bibinfo
  {school} {National University of Singapore} (\bibinfo {year} {2012}),\
  \bibinfo {note} {available at
  \url{http://scholarbank.nus.edu.sg/bitstream/handle/10635/35247/ZhuHJthesis.pdf}}\BibitemShut
  {NoStop}%
\bibitem [{\citenamefont {Zhu}(2014{\natexlab{b}})}]{Zhu14T}%
  \BibitemOpen
  \bibfield  {author} {\bibinfo {author} {\bibfnamefont {H.}~\bibnamefont
  {Zhu}},\ }\href@noop {} {\bibfield  {journal} {\bibinfo  {journal} {Phys.
  Rev. A}\ }\textbf {\bibinfo {volume} {90}},\ \bibinfo {pages} {032309}
  (\bibinfo {year} {2014}{\natexlab{b}})}\BibitemShut {NoStop}%
\bibitem [{\citenamefont {Klappenecker}\ and\ \citenamefont
  {R\"otteler}(2005)}]{KlapR05M}%
  \BibitemOpen
  \bibfield  {author} {\bibinfo {author} {\bibfnamefont {A.}~\bibnamefont
  {Klappenecker}}\ and\ \bibinfo {author} {\bibfnamefont {M.}~\bibnamefont
  {R\"otteler}},\ }in\ \href@noop {} {\emph {\bibinfo {booktitle} {IEEE
  International Symposium on Information Theory}}}\ (\bibinfo {address}
  {Adelaide, Australia},\ \bibinfo {year} {2005})\ pp.\ \bibinfo {pages} {1740
  --1744}\BibitemShut {NoStop}%
\bibitem [{\citenamefont {Dankert}(2005)}]{Dank05the}%
  \BibitemOpen
  \bibfield  {author} {\bibinfo {author} {\bibfnamefont {C.}~\bibnamefont
  {Dankert}},\ }\emph {\bibinfo {title} {Efficient Simulation of Random Quantum
  States and Operators}},\ \href@noop {} {\bibinfo {type} {Master thesis}},\
  \bibinfo  {school} {University of Waterloo} (\bibinfo {year} {2005}),\
  \bibinfo {note} {available online at
  \url{http://arxiv.org/abs/quant-ph/0512217}}\BibitemShut {NoStop}%
\bibitem [{\citenamefont {Dankert}\ \emph {et~al.}(2009)\citenamefont
  {Dankert}, \citenamefont {Cleve}, \citenamefont {Emerson},\ and\
  \citenamefont {Livine}}]{DankCEL09}%
  \BibitemOpen
  \bibfield  {author} {\bibinfo {author} {\bibfnamefont {C.}~\bibnamefont
  {Dankert}}, \bibinfo {author} {\bibfnamefont {R.}~\bibnamefont {Cleve}},
  \bibinfo {author} {\bibfnamefont {J.}~\bibnamefont {Emerson}}, \ and\
  \bibinfo {author} {\bibfnamefont {E.}~\bibnamefont {Livine}},\ }\href@noop {}
  {\bibfield  {journal} {\bibinfo  {journal} {Phys. Rev. A}\ }\textbf {\bibinfo
  {volume} {80}},\ \bibinfo {pages} {012304} (\bibinfo {year}
  {2009})}\BibitemShut {NoStop}%
\bibitem [{\citenamefont {Gross}\ \emph {et~al.}(2007)\citenamefont {Gross},
  \citenamefont {Audenaert},\ and\ \citenamefont {Eisert}}]{GrosAE07}%
  \BibitemOpen
  \bibfield  {author} {\bibinfo {author} {\bibfnamefont {D.}~\bibnamefont
  {Gross}}, \bibinfo {author} {\bibfnamefont {K.}~\bibnamefont {Audenaert}}, \
  and\ \bibinfo {author} {\bibfnamefont {J.}~\bibnamefont {Eisert}},\
  }\href@noop {} {\bibfield  {journal} {\bibinfo  {journal} {J. Math. Phys.}\
  }\textbf {\bibinfo {volume} {48}},\ \bibinfo {pages} {052104} (\bibinfo
  {year} {2007})}\BibitemShut {NoStop}%
\bibitem [{\citenamefont {Chau}(2005)}]{Chau05}%
  \BibitemOpen
  \bibfield  {author} {\bibinfo {author} {\bibfnamefont {H.~F.}\ \bibnamefont
  {Chau}},\ }\href@noop {} {\bibfield  {journal} {\bibinfo  {journal} {IEEE
  Trans. Inf. Theory}\ }\textbf {\bibinfo {volume} {51}},\ \bibinfo {pages}
  {1451} (\bibinfo {year} {2005})}\BibitemShut {NoStop}%
\bibitem [{\citenamefont {Amburg}\ \emph {et~al.}(2014)\citenamefont {Amburg},
  \citenamefont {Sharma}, \citenamefont {Sussman},\ and\ \citenamefont
  {Wootters}}]{AmbuSSW14}%
  \BibitemOpen
  \bibfield  {author} {\bibinfo {author} {\bibfnamefont {I.}~\bibnamefont
  {Amburg}}, \bibinfo {author} {\bibfnamefont {R.}~\bibnamefont {Sharma}},
  \bibinfo {author} {\bibfnamefont {D.~M.}\ \bibnamefont {Sussman}}, \ and\
  \bibinfo {author} {\bibfnamefont {W.~K.}\ \bibnamefont {Wootters}},\
  }\href@noop {} {\bibfield  {journal} {\bibinfo  {journal} {J. Math. Phys.}\
  }\textbf {\bibinfo {volume} {55}},\ \bibinfo {eid} {122206} (\bibinfo {year}
  {2014})}\BibitemShut {NoStop}%
\bibitem [{\citenamefont {Appleby}\ \emph
  {et~al.}(2014{\natexlab{b}})\citenamefont {Appleby}, \citenamefont
  {Bengtsson},\ and\ \citenamefont {Dang}}]{ApplBD14}%
  \BibitemOpen
  \bibfield  {author} {\bibinfo {author} {\bibfnamefont {D.~M.}\ \bibnamefont
  {Appleby}}, \bibinfo {author} {\bibfnamefont {I.}~\bibnamefont {Bengtsson}},
  \ and\ \bibinfo {author} {\bibfnamefont {H.~B.}\ \bibnamefont {Dang}},\
  }\href {http://arxiv.org/abs/1409.7987} {\enquote {\bibinfo {title} {{Galois
  Unitaries, Mutually Unbiased Bases, and MUB-balanced states}},}\ } (\bibinfo
  {year} {2014}{\natexlab{b}}),\ \Eprint {http://arxiv.org/abs/1409.7987}
  {arXiv:1409.7987} \BibitemShut {NoStop}%
\bibitem [{\citenamefont {Wootters}(1987)}]{Woot87}%
  \BibitemOpen
  \bibfield  {author} {\bibinfo {author} {\bibfnamefont {W.~K.}\ \bibnamefont
  {Wootters}},\ }\href@noop {} {\bibfield  {journal} {\bibinfo  {journal} {Ann.
  Phys.}\ }\textbf {\bibinfo {volume} {176}},\ \bibinfo {pages} {1} (\bibinfo
  {year} {1987})}\BibitemShut {NoStop}%
\bibitem [{\citenamefont {Dickson}(1958)}]{Dick58book}%
  \BibitemOpen
  \bibfield  {author} {\bibinfo {author} {\bibfnamefont {L.~E.}\ \bibnamefont
  {Dickson}},\ }\href@noop {} {\emph {\bibinfo {title} {Linear Groups, with an
  Exposition of the Galois Field Theory}}}\ (\bibinfo  {publisher} {Dover},\
  \bibinfo {address} {New York},\ \bibinfo {year} {1958})\BibitemShut {NoStop}%
\bibitem [{\citenamefont {Carmeli}\ \emph {et~al.}(2015)\citenamefont
  {Carmeli}, \citenamefont {Schultz},\ and\ \citenamefont {Toigo}}]{CarmST15}%
  \BibitemOpen
  \bibfield  {author} {\bibinfo {author} {\bibfnamefont {C.}~\bibnamefont
  {Carmeli}}, \bibinfo {author} {\bibfnamefont {J.}~\bibnamefont {Schultz}}, \
  and\ \bibinfo {author} {\bibfnamefont {A.}~\bibnamefont {Toigo}},\ }\href
  {http://arxiv.org/abs/1504.06415} {\enquote {\bibinfo {title} {{Covariant
  mutually unbiased bases}},}\ } (\bibinfo {year} {2015}),\ \Eprint
  {http://arxiv.org/abs/1504.06415} {arXiv:1504.06415} \BibitemShut {NoStop}%
\end{thebibliography}%

\end{document}